\documentclass[12pt,reqno]{amsart}

\usepackage{amscd,amssymb,amsmath,amsthm}
\usepackage[russian]{babel}
\usepackage[arrow,matrix]{xy}
\usepackage{graphicx}
\usepackage{color}
\usepackage{cite}
\newtheorem{thm}{Теорема}
\newtheorem{defn}{Определение}
\newtheorem{lemma}{Лемма}
\newtheorem{pro}{Утверждение}
\newtheorem{rk}{Замечание}
\newtheorem{cor}{Следствие}

\numberwithin{equation}{section} 

\newcommand{\bea}{\begin{eqnarray}}
\newcommand{\eea}{\end{eqnarray}}

\newcommand{\Q}{\mathbb{Q}}

\begin{document}
\title[$p$-адические меры Гиббса для модели ТС]
{$p$-адическая модель Твердых Сфер с тремя состояниями на дереве Кэли}

\author{О.\ Н.\ Хакимов}

\address{O.\ N.\ Khakimov\\ Институт математики, ул. Дурмон йули, 29, Ташкент, 100125,
Узбекистан.} \email {hakimovo@mail.ru}

\maketitle

\begin{abstract} В этой работе мы изучим $p$-адическую модель
(твердых сфер) ТС  с тремя состояниями на дереве Кэли.
При $k=2$ изучим трансляционно-инвариантные и периодические $p$-адические
меры Гиббса для модели ТС. Докажем, что при $p\neq2$ любая
$p$-адическая мера Гиббса является ограниченной. В частности,
будут показаны не существование сильного фазового перехода для модели ТС
на дереве Кэли порядка $k$.
\end{abstract}
\maketitle

{\bf{Ключевые слова:}} дерево Кэли, конфигурация, мера Гиббса,
модель TC, трансляционно-инвариантная мера, $p$-адические
числа.

\section{Определения и факты}

В работе \cite{MRS},\cite{rsh} были изучены вещественные
гиббсовские меры для модели ТС с тремя состояниями на
дереве Кэли порядка $k\geq1$. В этой работе мы изучим
$p$-адический аналог этой модели.

Известно, что $p$-адические модели в физике не
могут быть описаны, используя обычную теорию вероятностей \cite{24,34,48}. В
\cite{24} абстрактная $p$-адическая теория вероятностей была
развита посредством теории неархимедовых мер \cite{40}.
Вероятностные процессы на поле $p$-адических чисел были изучены
многими авторами (см. \cite{1a,3a,51,rh}). Не-архимедовый аналог
теоремы Колмогорова был доказан в \cite{16}.

Описание предельных мер Гиббса для данного гамильтониана является
одним из основных задач в теории гиббсовских мер. Полный анализ
множества таких мер довольно трудоемкий. По этой причине большая
часть работ по этой тематике посвящена изучению гиббсовских мер на
дереве Кэли \cite{grr,B,17, MRS,MNGM}.

В работе \cite{TMF2} был изучен $p$-адическая модель ТС с тремя
состояниями на дереве Кэли порядка $k$. Был доказан, что если
$k^2-4\not\equiv0(\operatorname{mod }p)$,
то существует единственная трансляционно-инвариантная
$p$-адическая мера Гиббса для модели ТС. В данной работе мы исследуем
случай $k^2-4\equiv0(\operatorname{mod }p)$. В этом случае будет
показана не единственность $p$-адических мер Гиббса для модели ТС.
Также, исследуем проблему ограниченности $p$-адических мер Гиббса
при любом $k$.

\subsection{$p$-адические числа и меры.} Каждое рациональное число
$x\neq 0$ может быть представлено в виде
$x=p^r\frac{n}{m}$, где $r,n\in \mathbb{Z}, m$-- положительное
число, $(n,m)=1$, причем $m$ и $n$ не делятся на $p$ и $p$ --
фиксированное простое число. $p$-Адическая норма $|x|_p$
определяется по формуле
$$
|x|_p= \left\{\begin{array}{ll}
p^{-r},& \text{ если }  x\neq 0,\\
0,& \text{ если } x=0.
 \end{array}\right.
$$

Эта норма удовлетворяет сильному неравенству треугольника:
$$
|x+y|_p\leq\max\{|x|_p,|y|_p\}.
$$
Это свойство показывает  неархимедовость нормы.

Из этого свойства непосредственно следуют следующие (свойства $p$-адической нормы):

1) если  $|x|_p\neq |y|_p$, то $|x-y|_p=\max\{|x|_p,|y|_p\}$;

2) если  $|x|_p=|y|_p$, то  $|x-y|_p\leq |x|_p$;

Пополнение поля рациональных чисел $\mathbb{Q}$ по  $p$-адической
норме приводит к полю $p$-адических чисел $\mathbb{Q}_p$ для
каждого простого $p$ \cite{29}.

Начиная с поля рациональных чисел $\mathbb{Q}$, мы можем получить
либо поле вещественных чисел $\mathbb{R}$, либо одно из полей
$p$-адических чисел $\mathbb{Q}_p$ (теорема Островского).

Каждое $p$-адическое число $x\neq0$ имеет единственное
каноническое разложение
\begin{equation}\label{ek}
x = p^{\gamma(x)}(x_0+x_1p+x_2p^2+\dots),
\end{equation}
где $\gamma=\gamma(x)\in \mathbb Z$ и $x_j$ целые числа, $0\leq
x_j \leq p - 1$, $x_0 > 0$, $j = 0, 1,2, ...$ (см
\cite{29,41,48}). В этом случае $|x|_p = p^{-\gamma(x)}$.

\begin{thm}\label{tx2} \cite{48} Уравнение
$x^2 = a$, $0\ne a =p^{\gamma(a)}(a_0 + a_1p + ...), 0\leq a_j
\leq p - 1$, $a_0 > 0$ имеет решение $x\in \Q_p$ тогда и только
тогда, когда выполняется слудующее:

1) $\gamma(a)$ четное;

2) $y^2\equiv a_0(\operatorname{mod} p)$ разрешимо, если $p\ne 2$;
$a_1=a_2=0$, если $p=2$.
\end{thm}

Множество $\mathbb Z_p=\{x\in\mathbb Q_p:\ |x|_p\leq1\}$
называется множеством целых $p$-адических чисел.
$\mathbb Z^*_p=\{x\in\mathbb Q_p:\ |x|_p\leq1\}$-- множество
$p$-адических единиц.

Следующая теорема известна как лемма Гензеля.

\begin{thm}\label{hl}\cite{29}
Пусть $F(x)=c_0+c_1x+\cdots+c_nx^n-$ многочлен с целыми
$p$-адическими коэффициентами, а
$F'(x)=c_1+2c_2x+3c_3x^2+\cdots+nc_nx^{n-1}-$ его производная.
Предположим, что $a_0\ -$ целое $p$-адическое число, для которого
$F(a_0)\equiv0(\operatorname{mod} p)$ а
$F'(a_0)\not\equiv0(\operatorname{mod} p)$. Тогда существует
единственное целое $p$-адическое число $a$, такое, что
$$F(a)=0\ \mbox{и}\ a\equiv a_0(\operatorname{mod} p).$$
\end{thm}

Для $a\in \mathbb{Q}_p$ и $r> 0$ обозначим
$$B(a, r) = \{x\in \mathbb{Q}_p : |x-a|_p < r\}.$$

$p$-адический {\it логарифм} определяется как ряд
$$\log_p(x) =\log_p(1 + (x-1)) =
\sum_{n=1}^\infty (-1)^{n+1}{(x-1)^n\over n},$$ который сходится
для  $x\in B(1, 1)$;  $p$-адическая экспонента определяется как
$$\exp_p(x) =\sum^\infty_{n=0}{x^n\over n!},$$
которая сходится для $x \in B(0, p^{-1/(p-1)})$.
\begin{lemma}\label{el} Пусть $x\in B(0, p^{-1/(p-1)}$. Тогда
$$|\exp_p(x)|_p = 1,\ \ |\exp_p(x)-1|_p = |x|_p, \ \ |\log_p(1 + x)|_p = |x|_p,$$
$$\log_p(\exp_p(x)) = x,\ \ \exp_p(\log_p(1 + x)) = 1 + x.$$
\end{lemma}

Более подробно об основах $p$-адического анализа и $p$-адической
математической физики можно найти в \cite{29,41,48}.

Пусть $(X,\mathcal B)$ измеримое пространство, где $\mathcal B$
алгебра подмножеств в $X$. Функция $\mu:\mathcal B\to\mathbb Q_p$
называется $p$-адической мерой, если для любого набора
$A_1,...,A_n\in\mathcal B$ такого, что $A_i\cap A_j=\varnothing,\
i\neq j$ имеет место
$$\mu\bigg(\bigcup_{j=1}^nA_j\bigg)=\sum_{j=1}^n\mu(A_j).$$
$p$-Адическая мера $\mu$ называется вероятностной, если $\mu(X)=1$ (см.
\cite{16}).
$p$-Адическая мера $\mu$ называется {\it ограниченной}, если
$\left\{\left|\mu(A)\right|_p: A\in\mathcal B\right\}<\infty$ (см. \cite{24}).
\subsection{Дерево Кэли}

Дерево Кэли $\Gamma^k=(V,L)$ порядка $k\geq1$ есть бесконечное
дерево (граф без циклов), из каждой вершины которого выходит ровно
$k+1$ ребер, $V\ -$ множество вершин и $L\ -$ множество ребер. Две
вершины $x$ и $y$ называются {\it ближайшими соседями}, если
существует ребро $l\in L$ соединяющее их и пишется как $l=\langle x,y\rangle$.
Расстояние $d(x,y)\ -$ число ребер кратчайшей пути, соединяюшей
$x$ и $y$.

Пусть $x^0\in V$ фиксированная точка. Введем обозначения:
$$W_n=\{x\in V: d(x,x^0)=n\},\qquad V_n=\bigcup_{m=1}^nW_m,$$

и
$$S(x)=\{y\in W_{n+1}: d(x,y)=1\},\quad x\in W_n.$$

\subsection{модель ТС}

Мы рассмотрим модель ТС с тремя состояниями на дерева Кэли. В этой
модели каждой вершине $x\in V$ ставится в соответствие одно из
значений $\sigma(x)\in\{0,1,2\}$. Значения $\sigma(x)\in\{1,2\}$
означают, что вершина $x\in V$ "занята"\, и $\sigma(x)=0$
означает, что вершина $x\in V$ "вакантна". Конфигурация
$\sigma=\{\sigma(x),\ x\in V\}$ на дереве Кэли есть функция из $V$
в $\{0,1,2\}$. Конфигурации в $V_n$ и $W_n$ определяются
аналогично.

Конфигурация $\sigma$
называется допустимой на дереве Кэли, если
$\sigma(x)+\sigma(y)\notin\{0,3\}$ для любой пары
ближайщих соседей $x$ и $y$ в $V$. Обозначим через $\Omega$ множество
всех допустимых конфигураций на дереве Кэли.

Для фиксированной
$\lambda=(\lambda_0,\lambda_1,\lambda_2)\in\mathbb Q_p^3$
определим $p$-адический гамильтониан модели ТС
\begin{equation}\label{Ham}
H_{\lambda}(\sigma)=\sum_{x\in V}\log_p\lambda_{\sigma(x)},\quad
\sigma\in\Omega.
\end{equation}

\section{построение $p$-адической меры гиббса.}
Мы построим $p$-адическую меру Гиббса для модели (\ref{Ham}). Так как в
определении $p$-адической меры Гиббса используется $\exp_p(x)$, то
все ниже следующие величины должны принадлежать множеству
$$\mathcal E_p=\{x\in\mathbb
Q_p:|x|_p=1,|x-1|_p<p^{-1/(p-1)}\}.$$\\
Как и в классическом случае, мы рассмотрим специальный класс меры
Гиббса.

Для $\sigma_n\in\Omega_{V_n}$ определим $\#\sigma_n=\sum_{x\in
V_n}{\bf 1}(\sigma_n(x)\geq1))$ (т.е., $\#\sigma_n$ число занятых
вершин в $\sigma_n$).

Пусть $z:x\to z_x=(z_{0,x},z_{1,x},z_{2,x})\in\mathcal E_p^3$
векторнозначная функция на $V$. Рассмотрим случай, когда
$\lambda_0=1,$ и $\lambda_1=\lambda_2=\lambda$. Для
$\lambda\in\mathcal E_p$ рассмотрим $p$-адическое вероятностное
распределение $\mu_z^{(n)}$ на $\Omega_{V_n}$, которое
определяется как
\begin{equation}\label{mu}
\mu_z^{(n)}(\sigma_n)=Z_{z,n}^{-1}\lambda^{\#\sigma_n}\prod_{x\in
W_n}z_{\sigma_n(x),x},\qquad n=1,2,...
\end{equation}
где $Z_{z,n}$ нормируюшая константа
\begin{equation}\label{norm}
Z_{z,n}=\sum_{\omega_n\in\Omega_{V_n}}\lambda^{\#\omega_n}\prod_{x\in
W_n}z_{\omega_n(x),x}.
\end{equation}

Говорят, что $p$-адическое вероятностное распределение $\mu^{(n)}$
согласовано, если $\mbox{для всех}\ n\geq1$ и
$\sigma_{n-1}\in\Omega_{V_{n-1}},$

\begin{equation}\label{us}
\sum_{\omega_n\in\Omega_{W_n}}\mu_z^{(n)}(\sigma_{n-1}\vee\omega_n)=\mu_z^{(n-1)}(\sigma_{n-1}).
\end{equation}
В этом случае по теореме Колмогорова \cite{16} существует
единственная мера $\mu_z$ на $\Omega$ такая, что
$\mu_z(\{\sigma\big|_{V_n}=\sigma_n\})=\mu_z^{(n)}(\sigma_n)$ для всех
$n$ и $\sigma_n\in\Omega_{V_n}$.

\begin{defn}
Мера $\mu_z^{(n)}$, определенная как (\ref{mu}) удовлетворяющая
(\ref{us}) называется $p$-адической мерой Гиббса для модели (\ref{Ham}),
соответствующей функции $z:x\in V\setminus\{x^0\}\to z_x$.
\end{defn}

Если существуют две $p$-адические меры Гиббса $\mu_z$ и $\mu_t$ такие,
что только одна из них является ограниченной, то говорят, что существует
{\it фазовый переход}. Более того, если существует последовательность множеств $\{A_n\}$
такая, что $A_n\in\Omega_{V_n}$ и $|\mu_z(A_n)|_p\to0$, $|\mu_t(A_n)|_p\to\infty$
при $n\to\infty$, то говорят, что существует {\it сильний
фазовый переход}. Если существуют две ограниченные $p$-адические меры Гиббса,
то говорят, что существует {\it квази фазовый переход} \cite{MNGM}.

Следующая теорема дает условие на $z_x$, гарантирующее
согласованность распределения $\mu_z^{(n)}$.

\begin{thm}
\cite{TMF2} Вероятностное распределение $\mu_z^{(n)},\ n=1,2,...$,
заданное формулой (\ref{mu}), согласованно тогда и только тогда,
когда для любого $x\in V$ имеют место следующие равенства:
\begin{equation}\label{usz}
z_{i,x}'=\lambda\prod_{y\in
S(x)}\frac{1+z_{i,y}'}{z_{1,y}'+z_{2,y}'},\qquad i=1,2
\end{equation}
где $z_{i,x}'=\lambda z_{i,x}/z_{0,x}\in\mathcal E_p,\ i=1,2$.
\end{thm}

\section{Трансляционно-инвариантная мера Гиббса}
Решение вида $z_x=(z_1,z_2)\in\mathcal E_p^2,\ x\neq x_0$ системы
уравнений (\ref{usz}) называется {\it трансляционно-инвариантным}.
Соответствующая $p$-адическая мера Гиббса
трансляционно-инвариантного решения системы уравнений (\ref{usz})
называется трансляционно-инвариантной мерой Гиббса.

Для того, чтобы найти трансляционно-инвариантные $p$-адические меры
Гиббса для модели ТС, рассмотрим следующие уравнения
\begin{equation}\label{wand}
z_{i}=\lambda\left(\frac{1+z_{i}}{z_{1}+z_{2}}\right)^k,\quad i=1,2.
\end{equation}
\begin{thm}\label{9} \cite{TMF2} $1)$
Пусть $p=2$. Если $k$ делится на $4$, то для модели (\ref{Ham})
существует единственная трансляционно-инвариантная $2$-адическая
мера Гиббса.\\
$2)$ Пусть $p\neq2$. Если $k^2-4$ не делится на $p$,
то для модели (\ref{Ham})
существует единственная трансляционно-инвариантная $p$-адическая
мера Гиббса.
\end{thm}
\begin{rk}
Условия Теоремы \ref{9} не являются необходимыми для единственности
трансляционно-инвариантных $p$-адических мер Гиббса \cite{TMF2}.
Возникает естественный вопрос: существует ли фазовый переход для
модели (\ref{Ham}) на дереве Кэли порядка $k$. Очевидно, что при $k=2$
условие Теоремы \ref{9} не выпольняется для любого простого числа $p$.
В работе \cite{TMF2} были показаны, что при $k=2$ и $p=3$ существуют три
трансляционно-инвариантные $p$-адические меры Гиббса.
\end{rk}
В этой работе мы исследуем трансляционно-инвариантные $p$-адические меры
Гиббса для модели (\ref{Ham}) на дереве Кэли порядка два.

\begin{pro}\label{pro1}
Пусть $k=2$ и $p>3$. Тогда система уравнений (\ref{wand}) имеет единственное
решение на инвариантном множестве $\left\{z\in\mathcal E_p^2: z_1=z_2\right\}$.
\end{pro}
\begin{proof}
Пусть $z_i=t,\ i=1,2$. Тогда из (\ref{wand}) получим
$$
4t^3-\lambda(t+1)^2=0.
$$
Функция $f(t)=4t^3-\lambda(t+1)^2$ является многочленом с целыми
$p$-адическими коэффициентами. Учитывая $\lambda\in\mathcal E_p$ и $p>3$ из
$f(1)=4(1-\lambda)$ и $f'(1)=8+4(1-\lambda)$ имеем
$f(1)\equiv0(\operatorname{mod }p)$
и $f(1)\not\equiv0(\operatorname{mod }p)$. В силу леммы Гензеля
существует единственное число $t^*\in\mathcal E_p$ такое, что $f(t^*)=0$.
Это означает, что функциональное уравнение (\ref{wand}) имеет единственное
решение $z^*=(t^*, t^*)$ на множестве $\left\{z\in\mathcal E_p^2: z_1=z_2\right\}$
\end{proof}
Обозначим $M_p=\{a\in\mathbb N: a\mbox{ квадратичный вычет по модулю }p\}$.
\begin{pro}\label{pro2}
Пусть $k=2$ и $p>3$. Если
$$\lambda\in\bigcup_{a\in M_p}\bigcup_{n\in\mathbb N}\left\{x\in\mathcal E_p:
\left|16x-16-3ap^{2n}\right|_p<p^{-2n}\right\},$$
то система уравнений (\ref{wand}) имеет два
решения на инвариантном множестве $\{z\in\mathcal E_p^2: z_1\neq z_2\}$.
\end{pro}
\begin{proof}
Вычитая второе уравнение (\ref{wand}) из первого, получим
$$(z_1-z_2)\bigg(1-\lambda\frac{(2+z_1+z_2)}{(z_1+z_2)^2}\bigg)=0$$
Так как $z_1\neq z_2$, то имеем
\begin{equation}\label{wti2}
(z_1+z_2)^2-\lambda(z_1+z_2)-2\lambda=0.
\end{equation}
Решив квадратное уравнение (\ref{wti2}), получим
\begin{equation}\label{z1+z2}
z_1+z_2=\frac{\lambda\pm\sqrt{\lambda(\lambda+8)}}{2}.
\end{equation}
Так как $\lambda\in\mathcal E_p$ и $p>3$, то имеем следующие
$$\lambda=1+\lambda_1p+\lambda_2p^2+\cdots$$
и
$$\lambda+8=9+\lambda_1p+\lambda_2p^2+\cdots$$

В силу Теоремы \ref{tx2}
существуют числа $\sqrt{\lambda}$ и $\sqrt{\lambda+8}$ в $\mathbb Q_p$.
С другой стороны, $z_1$ и $z_2$ должны удовлетворять
$\left|z_1+z_2-2\right|_p<1$. Заметив
$$\sqrt{\lambda(\lambda+8)}=3+\lambda'_1p+\lambda'_2p^2+\cdots$$
получим для $z_1+z_2=\frac{\lambda+\sqrt{\lambda(\lambda+8)}}{2},$
$$
\left|z_1+z_2-2\right|_p=\left|\lambda+\sqrt{\lambda(\lambda+8)}-4\right|_p=
\left|\left(\lambda_1+\lambda'_2\right)p+\cdots\right|_p<1
$$
и для $z_1+z_2=\frac{\lambda-\sqrt{\lambda(\lambda+8})}{2}$
$$
\left|z_1+z_2-2\right|_p=\left|\lambda-\sqrt{\lambda(\lambda+8)}-4\right|_p=
\left|-6+\left(\lambda'_1-\lambda'_2\right)p+\cdots\right|_p=1.
$$

Подставляя $z_1+z_2=\frac{\lambda+\sqrt{\lambda(\lambda+8)}}{2}$
в (\ref{wand}), мы получим
\begin{equation}\label{kv}
z=\bigg(\frac{2(1+z)}{\sqrt{\lambda}+\sqrt{\lambda+8}}\bigg)^2.
\end{equation}
Следовательно, получим решения квадратного уравнения (\ref{kv})
\begin{equation}\label{z1}
z^\pm=\frac{\bigg(\sqrt{\lambda}+\sqrt{\lambda+8}\bigg)
\bigg(2\sqrt{\lambda}\pm\sqrt{2\big(\lambda-4+
\sqrt{\lambda(\lambda+8)}\big)}\bigg)}{8}.
\end{equation}

Мы должны проверить существование
$\sqrt{2\left(\lambda-4+\sqrt{\lambda(\lambda+8)}\right)}$ в
$\mathbb Q_p$ и $z^\pm\in\mathcal E_p$.\\
В силу Теоремы \ref{tx2} число
$\sqrt{2\left(\lambda-4+\sqrt{\lambda(\lambda+8)}\right)}$
существует тогда и только тогда, когда существуют
$n\in\mathbb N,\ a\in M$ и
$\varepsilon\in\mathbb Z_p$ такие, что
$$
2\left(\lambda-4+\sqrt{\lambda(\lambda+8)}\right)=p^{2n}\left(a+\varepsilon p\right)
$$
Отсюда найдем
$$
\lambda=1+\frac{3a}{16}p^{2n}+\epsilon p^{2n+1},\qquad\mbox{где }\ |\epsilon|_p\leq1,
$$
которая эквивалентно $\left|16\lambda-16-3ap^{2n}\right|_p<p^{-2n}$.
Теперь проверим
$z^\pm\in\mathcal E_p$. Пусть $\left|16\lambda-16-3ap^{2n}\right|_p<p^{-2n}$
для некоторого натурального числа $n$ и $a\in M_p$. Тогда имеем
$$
\left|z^\pm-1\right|_p=\left|\bigg(\sqrt{\lambda}+\sqrt{\lambda+8}\bigg)
\bigg(2\sqrt{\lambda}\pm\sqrt{2\big(\lambda-4+
\sqrt{\lambda(\lambda+8)}\big)}\bigg)-8\right|_p=
$$
$$\left|(4+\alpha p)\left(2+\beta p\pm\gamma p^n\right)-8\right|_p<1,
\qquad\mbox{где }\ \alpha, \beta, \gamma\in\mathbb Z_p.
$$
Это означает, что $z^\pm\in\mathcal E_p$. Таким образом, мы доказали, что
функциональное уравнение (\ref{wand}) имеет две решения
$z^{(1)}=(z^+, z^-)$ и $z^{(2)}=(z^-, z^+)$ на множестве
$\{z\in\mathcal E_p^2: z_1\neq z_2\}$,
если $\left|16\lambda-16+3ap^{2n}\right|_p<p^{-2n}$.
\end{proof}

Из Утверждений \ref{pro1} и \ref{pro2} получим следующее
\begin{thm} Пусть $k=2$ и $p>3$. Тогда верны следующие утверждения:\\
$1)$ Если
$$\lambda\notin\bigcup_{a\in M_p}\bigcup_{n\in\mathbb N}\left\{x\in\mathcal E_p:
\left|16x-16-3ap^{2n}\right|_p<p^{-2n}\right\},$$
то существует единственная
трансляционно-инвариантная $p$-адическая мера Гиббса для модели (\ref{Ham});\\
$2)$ Если
$$\lambda\in\bigcup_{a\in M_p}\bigcup_{n\in\mathbb N}\left\{x\in\mathcal E_p:
\left|16x-16-3ap^{2n}\right|_p<p^{-2n}\right\},$$
то существуют три
трансляционно-инвариантные $p$-адические меры Гиббса для модели (\ref{Ham}).
\end{thm}

\section{периодическая мера гиббса}

В этом пункте мы исследуем периодические $p$-адические меры Гиббса
для модели (\ref{Ham}) и используем групповую структуру дерева Кэли.
Как известно (см. \cite{GR2}), что существует вазимно однозначное соответствие
между множеством вершин $V$ дерева Кэли порядка $k\geq1$ и группой $G_k$,
являющейся свободным произведением $k+1$ циклических групп второго порядка
с образующими
$a_1,a_2,\dots,a_{k+1}$.

\begin{defn}\cite{grr}
Пусть $\tilde{G}$ нормальная подгруппа группы $G_k$.
Множество $z=\left\{z_x: x\in G_k\right\}$ называется $\tilde{G}$-
периодическим, если $z_{yx}=z_x$ для любого $x\in G_k$ и $x\in\tilde{G}$.
Соответсвующая $p$-адическая мера Гиббса $\mu_z$ называется $\tilde{G}$-
периодической.
\end{defn}

Очевидно, что $G_k$-периодическая мера является
трансляционно-инвариантной. Обозначим
$$
G^{(2)}=\left\{x\in G_k: \mbox{длина слова } x\ \mbox{ четная}\right\}.
$$
Это множество является нормальной подгруппой индекса два \cite{GR2}.

Следующая теорема характеризует множество всех периодических $p$-адических
мер Гиббса для модели (\ref{Ham}).
\begin{thm}\label{char}
Пусть $\tilde{G}$ нормальная подгруппа конечного индекса в $G_k$.
Тогда любая
$\tilde{G}$-периодическая $p$-адическая мера Гиббса для модели (\ref{Ham})
являетсяс либо трансляционно-инвариантной, либо $G^{(2)}$-периодической.
\end{thm}
\begin{proof} Расмотрим функцию $F:\mathcal E^2_p\to\mathcal E^2_p$,
определенную как
$$F(z)=(F_1(z),F_2(z)),\quad\mbox{где }\
F_i(z)=\frac{1+z_i}{z_1+z_2},\ i=1,2.$$
Легко проверить, что $F_i(z)=F(t),\ i=1,2$
в том и только в том случае, если $z=t$. Следовательно, имеем
$F(z)=F(t)$ тогда
и только тогда,
когда $z=t$. Из этого свойства как в доказательстве Теоремы 2 в
\cite{MRS} следует, что любая
$\tilde{G}$-периодическая мера Гиббса является либо трансляционно-инвариантной
либо $G^{(2)}$-периодической.
\end{proof}

Благодаря этой теоремы имеем, что для того, чтобы найти периодические
(не трансляционно-инвариантные) меры Гиббса
для модели (\ref{Ham}), достаточно исследовать следующую систему уравнений:
\begin{equation}\label{per}
\left\{\begin{array}{ll}
z_1=\lambda\left(\frac{1+t_1}{t_1+t_2}\right)^k,\\
z_2=\lambda\left(\frac{1+t_2}{t_1+t_2}\right)^k,\\
t_1=\lambda\left(\frac{1+z_1}{z_1+z_2}\right)^k,\\
t_2=\lambda\left(\frac{1+z_2}{z_1+z_2}\right)^k,\\
z_1\neq t_1,\ z_2\neq t_2.
\end{array}\right.
\end{equation}
Мы рассмотрим (\ref{per}) при $k=2$. Предположим, что $z_1=z_2=z$. Тогда
из (\ref{per}) получим
$$
z=f(f(z)),\qquad\mbox{где } f(z)=\lambda\left(\frac{1+z}{2z}\right)^2.
$$
Заметим, что уравнение $f(f(z))-z=0$ содержит решение уравнения $f(z)-z=0$.
Но нас интересует только периодические (не являющиеся трансляционно-инвариантными)
решения. Поэтому рассмотрим уравнение
$$
\frac{f(f(z))-z}{f(z)-z}=0,
$$
которое эквивалентно
\begin{equation}\label{kv}
\lambda z^2-2(2-\lambda)z+\lambda=0.
\end{equation}
Это уравнение имеет решения в $\mathbb Q_p$
$$
z^\pm=\frac{2-\lambda\pm2\sqrt{1-\lambda}}{\lambda},
$$
если существует $\sqrt{1-\lambda}$ в $\mathbb Q_p$.

Для того, чтобы решении $z^\pm$ уравнения (\ref{kv}) были искомымы, надо
проверить $z^\pm\in\mathcal E_p$ и $f(z^\pm)-z^*\neq0$.
Сначало мы иследуем при каких $\lambda\in\mathcal E_p$ число
$\sqrt{1-\lambda}$ существует в $\mathbb Q_p$. Затем,
проверим $z^\pm\in\mathcal E_p$ и $f(z^\pm)-z^*\neq0$.
\begin{lemma}
Пусть $\lambda\in\mathcal E_p$. Число $\sqrt{1-\lambda}$ существует в $\mathbb Q_p$
тогда и только тогда, когда
\begin{equation}\label{lambda1}
\lambda\in\bigcup_{a\in M_p}\bigcup_{n\in\mathbb N}\left\{x\in\mathcal E_p:
\left|x-1+ap^{2n}\right|_p<p^{-2n}\right\},\quad\mbox{если } p>2
\end{equation}
и
\begin{equation}\label{lambda2}
\lambda\in\bigcup_{n\in\mathbb N}\left\{x\in\mathcal E_2:
\left|x-1+2^{2n}\right|_2<2^{-2n-2}\right\},\quad\mbox{если } p=2.
\end{equation}
\end{lemma}
\begin{proof}
Пусть $p=2$. Тогда из $\lambda\in\mathcal E_2$ получим
$$
\lambda=1+\lambda_22^2+\lambda_32^3+\cdots,\quad\mbox{где }\ \lambda_i\in\{0,1\},\ i=2,3,\dots.
$$
Отсюда, в силу Теоремы \ref{tx2} имеем
$$
1-\lambda=2^{2n}\left(1+\lambda'_{2n+3}2^3+\lambda'_{2n+4}2^4+\cdots\right),\qquad n\in\mathbb N,
$$
которое эквивалентно $\left|\lambda-1+2^{2n}\right|_2<2^{-2n-2}$.\\

Пусть $p>2$. Тогда из $\lambda\in\mathcal E_p$ получим
$$
\lambda=1+\lambda_1p+\lambda_2p^2+\cdots,\quad\mbox{где }\ \lambda_i\in\{0,1,\dots,p-1\},\ i=1,2,\dots.
$$
Тогда в силу Теоремы \ref{tx2} имеем
$$
1-\lambda=p^{2n}\left(a+\lambda'_{2n+1}p+\lambda'_{2n+2}p^2+\cdots\right),\quad n\in\mathbb N,\ a\in M_p.
$$
Следовательно, получим $\left|\lambda-1+ap^{2n}\right|_p<p^{-2n}$.
\end{proof}

Теперь проверим $z^\pm\in\mathcal E_p$. Заметив $|\lambda|_p=1$ и
$\left|1-\lambda\right|_p=p^{-2n}$ и используя свойство $p$-адической нормы, получим
$$
\left|z^\pm-1\right|_p=\frac{\left|2\left(1-\lambda\pm\sqrt{1-\lambda}\right)\right|_p}{|\lambda|_p}=
\left|2p^n\right|_p<p^{-1/(p-1)}.
$$
Это означает, что $z^\pm\in\mathcal E_p$.

Покажем $f(z^\pm)-z^\pm\neq0$.
$$
f(z^\pm)-z^\pm=\lambda\left(\frac{1+z^\pm}{2z^+}\right)^2-z^\pm=
\frac{-4\sqrt{1-\lambda}\left(1\pm\sqrt{1-\lambda}\right)^2}
{\lambda^2}.
$$
Так как $0<|1-\lambda|_p<1$, то имеем $\left|f(z^\pm)-z^\pm\right|_p\neq0$.
Следовательно, $z=(z^+,z^-), t=(z^-,z^+)$ и $z=(z^-,z^+), t=(z^+,z^-)$
являются решениями (\ref{per}) при $k=2$.

Таким образом, мы доказали следующее
\begin{pro} Пусть $k=2$. Тогда (\ref{per}) имеет по крайней мере два решения
на $\mathcal E^4_p$, если имеет место (\ref{lambda1}) и (\ref{lambda2}).
\end{pro}

Из этого утверждения получим следующую теорему
\begin{thm} Пусть имеет место (\ref{lambda1}) при $p>2$ (и (\ref{lambda2}) при $p=2$).
Тогда для модели (\ref{Ham}) существуют по крайней мере две периодические
$p$-адические меры Гиббса на дереве Кэли порядка два.
\end{thm}

\section{Ограниченность $p$-адических мер Гиббса}
В этом пункте мы будем исследовать ограниченности
$p$-адических мер Гиббса для модели (\ref{Ham}). Напомним, что $p$-адическая
вероятностная мера может быть неограниченной.

\begin{lemma}\label{rec} Пусть $\mu_z$ есть $p$-адическая мера Гиббса для модели (\ref{Ham}).
Тогда
для нормирующей константы $($\ref{norm}$)$ имеет место следующая
рекуррентная формула:
$$Z_{z,n+1}=A_{z,n}Z_{z,n},\qquad n=1,2,\dots,$$
где $A_{z,n}$ определяется по формуле $($\ref{A_n}$)$.
\end{lemma}
\begin{proof}
Пусть функция $z':x\to z_x'=(z_{1,x}',z_{2,x}')\in\mathcal E_p^2$
удовлетворяет функционльному уравнению (\ref{usz}). Тогда для любого
$z_{0,x}\in\mathcal E_p,\ x\in V$ существует функция $a_z(x)$ такая, что
\begin{equation}\label{a(x)}
\begin{array}{ll}
\prod_{y\in S(x)}\left(\lambda z_{1,y}+\lambda z_{2,y}\right)=a_z(x)z_{0,x},\\[2mm]
\prod_{y\in S(x)}\left(z_{0,y}+\lambda z_{i,y}\right)=a_z(x)z_{i,x},\ i=1,2
\end{array}
\end{equation}
где $z_{i,x}=z_{0,x}z'_{i,x}/\lambda,\ i=1,2$. Тогда для любой конфигурации
$\sigma\in\Omega_{V_{n}}$ имеет место следующие
$$
\prod_{x\in W_n}\prod_{y\in S(x)\atop \sigma(x)=0}(\lambda z_{1,y}+\lambda z_{2,y})
\prod_{y\in S(x)\atop \sigma(x)=1}(z_{0,y}+\lambda z_{1,y})
\prod_{y\in S(x)\atop \sigma(x)=2}(z_{0,y}+\lambda z_{2,y})=
$$
\begin{equation}\label{A_n}
\prod_{x\in W_n}a_z(x)z_{\sigma(x),x}=A_{z,n}\prod_{x\in W_{n}}z_{\sigma(x),x},
\qquad\mbox{где }A_{z,n}=\prod_{x\in W_n}a_z(x).
\end{equation}
Учитывая (\ref{mu}) и (\ref{norm}) из (\ref{A_n}) получаем
$$
1=\sum_{\sigma\in\Omega_{V_{n}}}\sum_{\omega\in\Omega_{W_{n+1}}}\mu_z^{(n+1)}(\sigma\vee\omega)=
\sum_{\sigma\in\Omega_{V_{n}}}\sum_{\omega\in\Omega_{W_{n+1}}}\frac{1}{Z_{z,n+1}}\lambda^{\#\sigma\vee\omega}
\prod_{x\in{W_{n+1}}}z_{\omega(x),x}=
$$
$$
\frac{A_{z,n}}{Z_{z,n+1}}\sum_{\sigma\in\Omega_{V_n}}\lambda^{\#\sigma}\prod_{x\in{W_n}}z_{\sigma(x),x}=
\frac{A_{z,n}}{Z_{z,n+1}}Z_{z,n}.
$$
Отсюда, $Z_{z,n+1}=A_{z,n}Z_{z,n}$.
\end{proof}
\begin{thm} $p$-адическая
мера Гиббса для модели (\ref{Ham}) является ограниченной тогда и только
тогда, когда $p\neq2$.
\end{thm}
\begin{proof}
Пусть $z'=(z_{1,x}', z_{2,x}')\in\mathcal E_p^2$ решение
функционального уравнения (\ref{usz}) и $\mu_z$ $p$-адическая мера Гиббса
соответсвующей функции
$z=\left(z_{0,x}, \frac{z_{0,x}z_{1,x}'}{\lambda}, \frac{z_{0,x}z_{2,x}'}{\lambda}\right)$,
где $z_{0,x}\in\mathcal E_p$. Тогда в силу Леммы \ref{rec} при всех $n\geq1$ имеем
$$
Z_{z,n}=\prod_{x\in V_{n-1}}a_z(x),\qquad\mbox{где }\ a_z(x)=z_{0,x}^{k-1}(z_{1,x}'+z_{2,x}')^k.
$$
Так как $z_{0,x}, z_{1,x}', z_{2,x}'\in\mathcal E_p$, то имеем
$$
\left|a_z(x)\right|_p=\left\{\begin{array}{ll}
1, & \mbox{если }p\neq2,\\
2^{-k}, & \mbox{если }p=2.
\end{array}\right.
$$
Следовательно,
$$
\left|Z_{z,n}\right|_p=\left\{\begin{array}{ll}
1, & \mbox{если }p\neq2,\\
2^{-k|V_{n-1}|}, & \mbox{если }p=2.
\end{array}\right.
$$
Отсюда для любой
конфигурации $\sigma\in\Omega$ получим
$$
\left|\mu_z^{(n)}(\sigma)\right|_p=\frac{\left|\lambda^{\#\sigma}
\prod_{x\in W_n}z_{\sigma(x),x}\right|_p}{\left|Z_{z,n}\right|_p}=
\left\{\begin{array}{ll}
1, & \mbox{если }p\neq2,\\
2^{k|V_{n-1}|}, & \mbox{если }p=2.
\end{array}\right.
$$
Это означает, что мера $\mu_z$ является ограниченной тогда и только тогда,
когда $p\neq2$.
\end{proof}
\begin{cor}
Для модели (\ref{Ham}) не существует фазового перехода. В частности, не
существует сильного фазового перехода.
\end{cor}
\begin{cor}
Пусть $k=2$ и $p>3$. Если
$$\lambda\in\bigcup_{a\in M_p}\bigcup_{n\in\mathbb N}\left\{x\in\mathcal E_p:
\left|16x-16-3ap^{2n}\right|_p<p^{-2n}\right\},$$
то для модели (\ref{Ham}) существует квази фазовый переход.
\end{cor}

\end{document}